\documentclass[11pt,english]{article}
\usepackage[latin9]{inputenc}
\usepackage{float}
\usepackage{amsmath}
\usepackage{amsthm}
\usepackage{amssymb}
\usepackage[numbers]{natbib}

\makeatletter

\providecommand{\tabularnewline}{\\}

  \theoremstyle{plain}
  \newtheorem{thm}{\protect\theoremname}[section]
   \theoremstyle{plain}
  
  \newtheorem{question}{\protect\questionname}
  \theoremstyle{definition}
  \newtheorem{example}{\protect\examplename}[section]
  \theoremstyle{remark}
  \newtheorem{rem}{\protect\remarkname}[section]
  \theoremstyle{plain}
  \newtheorem{lem}{\protect\lemmaname}[section]
  \theoremstyle{plain}
  \newtheorem{cor}{\protect\corollaryname}[section]
  \theoremstyle{plain}
  
  \theoremstyle{definition}
  \newtheorem{definition}{\protect\definitionname}[section]

\usepackage[all,2cell]{xy}  \UseAllTwocells \SilentMatrices
\usepackage{hyperref}
\usepackage{amsmath}
\usepackage{slashed}
\usepackage{graphicx}

\usepackage[misc,geometry]{ifsym} 

\hypersetup{colorlinks=true, linkcolor=blue, citecolor=blue, filecolor=blue, pagecolor=blue, urlcolor=red, pdfauthor={}, pdftitle={}}

\makeatother

\usepackage{babel}
  \providecommand{\definitionname}{Definition}
  \providecommand{\examplename}{Example}
  \providecommand{\lemmaname}{Lemma}
  \providecommand{\propositionname}{Proposition}
  \providecommand{\questionname}{Question}
  \providecommand{\remarkname}{Remark}
\providecommand{\corollaryname}{Corollary}
\providecommand{\theoremname}{Theorem}

\begin{document}

\title{An Axiomatization Proposal and a Global Existence Theorem for Strong Emergence Between Parameterized Lagrangian Field Theories}

\author{Yuri Ximenes Martins\footnote{yurixm@ufmg.br (corresponding author)}\,  and  Rodney Josu\'e Biezuner\footnote{rodneyjb@ufmg.br} }
\maketitle
\noindent \begin{center}
\textit{Departamento de Matem\'atica, ICEx, Universidade Federal de Minas Gerais,}  \\  \textit{Av. Ant\^onio Carlos 6627, Pampulha, CP 702, CEP 31270-901, Belo Horizonte, MG, Brazil}
\par\end{center}
\begin{abstract}
In this paper we propose an axiomatization for the notion of strong emergence
phenomenon between field theories depending on additional parameters, which we call parameterized field theories. We present
sufficient conditions ensuring the existence of such phenomena between
two given Lagrangian theories. More precisely,
we prove that a Lagrangian field theory depending linearly on an additional parameter emerges from every multivariate polynomial theory evaluated at differential operators which have well-defined Green functions (or, more generally, that has a right-inverse in some extended sense). As a motivating example, we show that the phenomenon of gravity emerging from noncommutativy, in the context of a real or complex scalar field theory, can be recovered from our emergence theorem. We also show that, in the same context, we could also expect the reciprocal, i.e., that noncommutativity could emerge from gravity. Some other particular cases are analyzed.
\end{abstract}

\section{Introduction \label{sec_introduction}}

$\quad\;\,$The term \emph{emergence phenomenon} has been used for
years in many different contexts. In each of them, Emergence Theory
is the theory which studies those kinds of phenomena. E.g, we have
versions of it in Philosophy, Art, Chemistry and Biology \citep{0,1,2}.
The term is also often used in Physics, with different meanings
(for a review on the subject, see \citep{3,4}. For an axiomatization
approach, see \citep{5}). This reveals that the concept of emergence
phenomenon is very general and therefore difficult to formalize. Nevertheless,
we have a clue of what it really is: when looking at all those instances of the phenomenon
we see that each of them is about describing a system in terms of other system, possibly in different \emph{scales}. Thus, an emergence
phenomenon is about a relation between two different systems, the
\emph{emergence relation}, and a system emerges from another when
it (or at least part of it) can be recovered in terms of the other
system, which is presumably more fundamental, at least in some scale.
The different emergence phenomena in Biology, Philosophy, Physics,
and so on, are obtained by fixing in the above abstract definition
a meaning for system, scale, etc.

Notice that, in this approach, in order to talk about emergence we
need to assume that to each system of interest we have assigned a
\emph{scale}. In Mathematics, scales are better known as \emph{parameters}.
So, emergence phenomena occur between \emph{parameterized
systems}. This kind of assumption (that in order to fix a system we
have to specify the scale in which we are considering it) is at the
heart of the notion of effective field theory, where the scale (or
parameter) is governed by Renormalization Group flows \citep{6,7,8,9}.
Notice, in turn, that if a system emerges from another, then the
second one should be more fundamental, at least in the scale (or parameter)
in which the emergence phenomenon is observed. This also puts Emergence
Theory in the framework of searching for the fundamental theory of
Physics (e.g Quantum Gravity), whose systems should be the minimal
systems relative to the emergence relation \citep{3,4}. The main
problem in this setting is then the existence problem for the minimum.
A very related question is the general existence problem: \emph{given
two systems, is there some emergence relation between them?}

One can work on the existence problem at different levels of depth.
Indeed, since the systems is question are parameterized one can ask
if there exists a correspondence between them in \emph{some scales}
or in \emph{all scales}. Obviously, requiring a complete correspondence
between them is much stronger than requiring a partial one. On
the other hand, in order to attack the existence problem we also have
to specify which kind of emergence relation we are looking for. Again,
is it a full correspondence, in the sense that the emergent theory
can be fully recovered from the fundamental one, or is it only a 
partial correspondence, through which only certain aspects can be
recovered? Thus, we can say that we have the following four versions
of the existence problem for emergence phenomena:

\begin{table}[H]
\begin{centering}
\begin{tabular}{|c|c|c|c|c|}
\cline{2-5} 
\multicolumn{1}{c|}{} & \emph{weak} & \emph{weak-scale} & \emph{weak-relation} & \emph{strong}\tabularnewline
\hline 
\emph{relation} & partial & full & partial & full\tabularnewline
\hline 
\emph{scales} & some & some & all & all\tabularnewline
\hline 
\end{tabular}
\par\end{centering}
\caption{Types of Emergence}
\end{table}

In Physics one usually works on finding weak emergence phenomena.
Indeed, one typically shows that certain properties of a system can
be described by some other system at some limit, corresponding to
a certain regime of the parameter space. These emergence phenomena
are strongly related with other kind of relation: the \emph{physical
duality}, where two different systems reveal the same physical properties.
One typically builds emergence from duality. For example, AdS/CFT
duality plays an important role in describing spacetime geometry (curvature)
from mechanic statistical information (entanglement entropy) of dual
strongly coupled systems \citep{10,11,12,13,14,15}. 

There are also some interesting examples of weak-scale emergence relations,
following again from some duality. These typically occur when the
action functional of two Lagrangian field theories are equal at some
limit. The basic example is gravity emerging from noncommutativity
following from the duality between commutative and noncommutative
gauge theories established by the Seiberg-Witten maps \citep{16}.
Quickly, the idea was to consider a gauge theory $S[A]$ and modify
it into two different ways: 
\begin{enumerate}
\item by considering $S[A]$ coupled to some background field $\chi$, i.e,
$S_{\chi}[A;\chi]$;
\item by using the Seiberg-Witten map to get its noncommutative analogue
$S_{\theta}[\hat{A};\theta]$.
\end{enumerate}
$\quad\;\,$Both new theories can be regarded as parameterized theories:
the parameter (or scale) of the first one is the background field
$\chi$, while that of the second one is the noncommutative parameter
$\theta^{\mu\nu}$. By construction, the noncommutative theory $S_{\theta}[A;\theta]$
can be expanded in a power series on the noncommutative parameter,
and we can also expand the other theory $S_{\chi}[A;\chi]$ on the
background field, i.e, one can write 
\[
S_{\chi}[A;\chi]=\sum_{i=0}^{\infty}S_{i}[A;\chi^{i}]=\lim_{n\rightarrow\infty}S_{(n)}[A;\chi]\quad\text{and}\quad S_{\theta}[\hat{A};\theta]=\sum_{i=0}^{\infty}S_{i}[A;\theta^{i}]=\lim_{n\rightarrow\infty}S_{(n)}[A;\theta],
\]
where $S_{(n)}[A;\chi]=\sum_{i=0}^{n}S_{i}[A;\chi^{i}]$ and $S_{(n)}[A;\theta]=\sum_{i=0}^{n}S_{i}[A;\theta^{i}]$
are partial sums. One then tries to find solutions for the following
question:
\begin{question}
\label{question_int}Given a gauge theory $S[A]$, is there a background
version $S_{\chi}[A;\chi]$ of it and a number $n$ such that for
every given value $\theta^{\mu\nu}$ of the noncommutative parameter
there exists a value of the background field $\chi(\theta)$, possibly
depending on $\theta^{\mu\nu}$, such that for every gauge field $A$
we have $S_{(n)}[A;\chi(\theta)]=S_{(n)}[A;\theta]$?
\end{question}
Notice that if rephrased in terms of parameterized theories, the question
above is precisely about the existence of a weak-scale emergence between
$S_{\chi}$ and $S_{\theta}$, at least up to order $n$. This can
also be interpreted by saying that, in the context of the gauge theory
$S[A]$, the background field $\chi$ emerges in some regime from
the noncommutativity of the spacetime coordinates. Since the noncommutative
parameter $\theta^{\mu\nu}$ depends on two spacetime indexes, it
is suggestive to consider background fields of the same type, i.e, 
$\chi^{\mu\nu}$. In this case, there is a natural choice: metric
tensors $g^{\mu\nu}$. Thus, in this setup, the previous question
is about proving that in the given gauge context, gravity emerges
from noncommutativity at least up to a perturbation of order $n$.
This has been proved to be true for many classes of gauge theories and
for many values of $n$ \citep{17,18,19,20,21}. On the other hand,
this naturally leads to other two questions:
\begin{enumerate}
\item Can we find some emergence relation between gravity and noncommutativity
in the nonperturbative setting? In other words, can we extend the
weak-scale emergence relation above to a strong one? 
\item Is it possible to generalize the construction of the cited works to
other kinds of background fields? In other words, is it possible to
use the same idea in order to show that different fields emerge from
spacetime noncommutativity? Or, more generally, is it possible to build a version of it for some general class of field theories?
\end{enumerate}
$\quad\;\,$The first of these questions is about finding a strong
emergence phenomena and it has a positive answer in some cases \citep{22,23,24,25}.
The second one, in turn, is about finding systematic and
general conditions ensuring the existence (or nonexistence) of emergence
phenomena. At least to the authors knowledge, there are no such general studies,
specially focused on the strong emergence between field theories.
It is precisely this point that is the focus of the present work.
Indeed we will:
\begin{enumerate}
\item based on Question \ref{question_int}, propose an axiomatization for
the notion of \emph{strong emergence} between field theories;
\item establish sufficient conditions ensuring that a given Lagrangian field
theory emerges from each theory belonging to a certain class of theories.
\end{enumerate}

We will work on the setup of \textit{parameterized field theories}, which are given by families $S[\varphi;\varepsilon]$ of field theories depending on a \textit{fundamental parameter $\varepsilon$}. In the situations described above, $\varepsilon$ is the the noncommutative parameter $\theta^{\mu\nu}$ or the background field $\chi^{\mu\nu}$. In the cases where the emergence was explicitly obtained, it was of summary importance that the parameters $\theta^{\mu\nu}$ and $\chi^{\mu\nu}$ belongs to the \textit{same} class of fields and that the corresponding action functions are defined on the \textit{same} field $A_{\mu}$. Thus, given two parameterized theories $S[\varphi;\varepsilon]$ and $S'[\varphi;\varepsilon']$ we always assume that they are defined on the same fields and that the parameters $\varepsilon$ and $\varepsilon'$ belong to the same space. Keeping Question \ref{question_int} as a motivation, let us say that $S[\varphi;\varepsilon]$ \textit{emerges} from $S'[\varphi;\varepsilon']$ if there is a map $F$ on the space of parameters such that, for every $\varepsilon$ and every field $\varphi$ we have $S[\varphi;\varepsilon]=S'[\varphi;F(\varepsilon)]$. We call $F$ a \textit{strong emergence phenomenon} between $S$ and $S'$.

Our main result states that for certain $S[\varphi;\varepsilon]$ and $S'[\varphi;\varepsilon']$, depending on $\varepsilon$ and $\varepsilon'$ in a suitable parameter space, in the sense that it has some special algebraic structure, then these strong emergence phenomena exist. The formal statement of this result will be presented in Section \ref{sec_emergence_thm} after some technical  digression. But, as a motivation, let us state a particular version of it and show how it can be used to recover an example of emergence between gravity and noncommutativity.

First of all, recall that the typical field theory has a kinetic part and an interacting part. The kinetic part is usually quadratic and therefore of the form $\mathcal{L}_{\operatorname{knt},i}(\varphi_i)=\langle\varphi_i,D_i\varphi_i\rangle_i$, with $i=1,...,N$, where $N$ is the number of fields, $D_i$ are differential operators and $\langle\cdot,\cdot\rangle_i$ are pairings in the corresponding space of fields. Summing the space of fields and letting $\varphi=(\varphi_1,...,\varphi_N)$, $D=\oplus_iD_i$ and $\langle\cdot,\cdot\rangle=\oplus_i \langle\cdot,\cdot\rangle_i$, one can write the full kinetic part as $\mathcal{L}_{\operatorname{knt}}(\varphi)=\langle\varphi,D\varphi \rangle$.  On the other hand, the interacting part is typically polynomial, i.e., it is of the form $\mathcal{L}_{\operatorname{int}}(\varphi)=\langle \varphi,p^l[D_1,...,D_N]\varphi\rangle$, where $l\geq0$ is the degree of the interaction and $p^l[D_1,...,D_N]=\sum_{\vert\alpha\vert\leq l}f_{\alpha}\cdot D^{\alpha}$. Since the space of differential operators constitutes an algebra, it follows that $p^l[D_1,...,D_N]$ is a differential operator too, so that both the kinetic and the interacting parts (and therefore the sum of them, which constitute the typical lagrangians) are of the form $\mathcal{L}(\varphi)=\langle\varphi,D\varphi\rangle$. 

Notice, furthermore, that the pairing $\langle \cdot, \cdot \rangle$ is typically induced by fixed geometric structures (such as metrics) in the spacetime manifold $M$ and in the field bundle $E$. Thus, in the parameterized context the natural dependence on the parameter is on the differential operator, i.e., the typical parameterized Lagrangian theories are of the form $\mathcal{L}(\varphi;\varepsilon)=\langle \varphi, D_{\varepsilon}\varphi \rangle$. In the case of polynomial theories (e.g., those describing interations), it is more natural to assume that the dependence on $D_\varepsilon$ is actually on the coefficient functions $f_{\alpha}$, i.e., $p^l_\varepsilon[D_1,...,D_N]=\sum_{\vert \alpha \vert \leq l}f_{\alpha}(\varepsilon)D^\alpha$. These coefficient functions could be scalar functions or, more generally, parameter-valued functions if we have an action of parameters in differential operators. In this last case, the polynomial is  $p^l_\varepsilon[D_1,...,D_N]=\sum_{\vert \alpha \vert \leq l}f_{\alpha}(\varepsilon)\cdot D^\alpha$, where the dot denotes the action of parameters in differential operators. $\underset{\;}{\underset{\;}{\;}}$

\noindent \textbf{Emergence Theorem} (rough version)
\textit{Let $M$ be a spacetime manifold and $E\rightarrow M$ a field bundle which define a pairing $\langle \cdot ,\cdot \rangle$ in $E$. Let $\mathcal{L}_1(\varphi;\varepsilon)=\langle \varphi, D_{\varepsilon}\varphi \rangle$ be an arbitrary parameterized theory and $\mathcal{L}_2(\varphi;\delta)=\langle \varphi,p^l_{\delta}[D_1,...,D_N]\varphi\rangle$ a polynomial theory, e.g., an interaction term. Suppose that:
\begin{enumerate}
    \item the parameters $\varphi$ are nonnegative real numbers or, more generally, positive semi-definite operators or tensors acting on the space of differential operators;
    \item the dependence of $D_\varepsilon$ on $\varepsilon$ is of the form $D_\varepsilon \varphi = \varepsilon \cdot D\varphi$, where $D$ is a fixed differential operator;
    \item the coefficient $f_\alpha(\delta)$ of $p^l_{\delta}[D_1,...,D_N]$ are nowhere vanishing scalar or parameter-valued functions, and the differential operators $D_1,...,D_N$ have well-defined Green functions.
\end{enumerate}
Then the theory $\mathcal{L}_1(\varphi;\varepsilon)$ emerges from the theory $\mathcal{L}_2(\varphi;\varepsilon)$.$\underset{\;}{\underset{\;}{\;}}$}

Now, looking at \cite{17}, consider again the question of emergent gravity in a four dimensional spacetime. The first order perturbation of a real scalar field theory $\varphi$, coupled to a semi-Riemannian\footnote{In \cite{17} the author considered only Lorentzian spacetimes. However, we notice that in order to get the emergence phenomena the Lorentzian signature was not explicitly used, so that the same holds in the semi-Rimennanian case.} gravitational field $g=\eta+h$, in an abelian gauge background, is given by equation (10) of \cite{17}\footnote{In \cite{17} an expansion of the form $g_{\mu\nu}=\eta_{\mu\nu} +h_{\mu\nu} +h\eta_{\mu\nu}$ was considered, where $h$ is some function, but the author concluded that the emergence phenomenon exists only if $h=0$. Thus, we are assuming this necessary condition from the beginning.}:
\begin{align}
    \mathcal{L}_{\operatorname{grav}}(\varphi;h)&=\partial^{\mu}\varphi\partial_{\mu} \varphi - h^{\mu\nu}\partial_\mu \varphi \partial_\nu \varphi \\
     &= - \langle \varphi, \square\varphi \rangle  + \langle \varphi, (h \cdot D_1 )\varphi\rangle \\
     &= \mathcal{L}_{0}(\varphi)+ \mathcal{L}_{1}(\varphi;h) \label{grav_lagrangian}.
\end{align}
where in the second step we did integration by parts and used the fact that the spacetime manifold is boundaryless. Furthermore, the pairing is just  $\langle\varphi,\varphi'\rangle = \varphi\varphi'$, while $D_1=\partial_\mu\partial_\nu$ and for every 2-tensor $t=t^{\mu\nu}$, we have the trace action $t\cdot D=t^{\mu\nu}\partial_{\mu}\partial_\nu$. In particular, $\eta\cdot D_1 = \partial^\mu\partial_\nu=\square$. On the other hand, the first order expansion in $\theta$ of the Seiberg-Witten dual of a real scalar field theory, in an abelian gauge background, is given by (equation (7) of \cite{17}):
\begin{align}
  \mathcal{L}_{\operatorname{ncom}}(\varphi;\theta)& = \partial^{\mu}\varphi\partial_{\mu} \varphi + 2\theta^{\mu\alpha}F_{\alpha\kappa}\eta^{\kappa\nu}(-\partial_\mu\varphi \partial_\nu \varphi  + \frac{1}{4} \eta_{\mu\nu} \partial^{\rho}\varphi \partial_{\rho}\varphi) \\
  &=  \langle \varphi, (- \square) \varphi \rangle +  \langle \varphi, (f(\theta)\cdot D_2) \varphi \rangle \\
  &= \langle \varphi, p^1_{\theta}[\square,D_2]\varphi \rangle \label{noncomm_lagrangian},
\end{align}
where in the first step we again integrated by parts, and used the fact that $p^1_{\theta}[x,y]$ is the first order polynomial $(-1)\cdot x + f(\theta)\cdot y$, where $f(\theta)=\theta$ and, in coordinates, $D_2= 2F_{\alpha\kappa}\eta^{\kappa\nu}(\partial_\mu\partial_\nu-\frac{1}{4}\eta_{\mu\nu}\square)$. 

Notice that if $h^{\mu\nu}$ is a positive-definite tensor (which is the case in \textit{Riemannian signature}), then $\mathcal{L}_1(\varphi;h)$ in (\ref{grav_lagrangian}) satisfies the hypotheses of the emergence theorem. On the other hand, since in the Riemannian setting $\square$ is the Laplacian and $D_2$ is essentially a combination of generalized Laplacians, both of them have Green functions\footnote{Actually, they are elliptic and thus have Fredholm inverses; this is enough for us.}. Thus, if we forget the trivial case $\theta^{\mu\nu}=0$, then $\mathcal{L}_{\operatorname{ncom}}(\varphi;\theta)$ also satisfies the hypotheses of the emergence theorem. Thus, as a consequence we see that \textit{the gravitational term $\mathcal{L}_1(\varphi;h)$ emerges from the noncommutative theory $\mathcal{L}_{\operatorname{ncom}}(\varphi;\theta)$}, as also proved in \cite{17}. In other words, there is a function $F$ on the space of 2-tensors, such that for every $h^{\mu\nu}$ we have $\mathcal{L}_1(\varphi;h^{\mu\nu})=\mathcal{L}_{\operatorname{ncom}}(\varphi;F(h^{\mu\nu})).$

Some comments.
\begin{enumerate}
    \item In \cite{17} the author proved \textit{explicitly} that gravity emerges from noncommutativy. More precisely, he gave an explicit expression for the function $F(h^{\mu\nu})$. Here, however, our result is only about the existence of such function,
    \item While above we had to assume Rimemannian signature, our main result (Theorem \ref{main_theorem}) applies equally well to the Lorentzian signature. It is different, however, of the rough version stated above.
    \item In \cite{17} analogous emergence phenomena were established in the case of complex scalar fields. Our main theorem also holds for complex fields.
    \item Notice that the scalar theory in the gravitational background  (\ref{grav_lagrangian}) can also be written as $\mathcal{L}_{\operatorname{grav}}(\varphi;h)=\langle\varphi,p^1_h[\square,D_1]\varphi \rangle$, where $p^1_h[x,y]$ is the first order polynomial $(-1)\cdot x + f(h)\cdot y$, where, again $f(h)=h$. Thus, we can also see the gravitational theory as a polynomial theory, defined by the same polynomial, but evaluated in different differential operators. On the other hand, (\ref{noncomm_lagrangian}) can be written as $\mathcal{L}_0(\varphi)+\mathcal{L}_2(\varphi;\theta)$, where $\mathcal{L}_2(\varphi;\theta)=\langle \varphi, (\theta\cdot D_2) \varphi \rangle$. Since $\square$ and $D_1$ are generalized Laplacians, they have Green functions. Thus, if one restricts to the case of noncommutativity parameters $\theta$ which are positive definite (say, e.g., positive definite symplectic forms), then one can again use our emergence theorem to conclude the reciprocal fact: that  $\mathcal{L}_2(\varphi;\theta)$ emerges from $\mathcal{L}_{\operatorname{grav}}(\varphi;h)$, i.e., that \textit{noncommutativity may also emerges from gravity}. 
\end{enumerate}

This work is organized as follows. In Section \ref{sec_1} we propose, based on the lines of this introduction, a formal definition for the notion of strong emergence between parameterized field theories, and we introduced the \textit{emergence problem} which is about determining if there is some emergence phenomenon between such two given theories. We discuss why it is natural (or at least reasonable) to restrict to a certain class of parameterized field theories, defined by certain ``generalized operators''. In Section \ref{sec_emergence_thm} our main theorem is stated. We begin by introducing its ``syntax'', leaving the precise formal statement (or ``semantic'') to Section \ref{sec_formal_statement}. Before giving the proof, which is a bit technical and based on induction arguments, in Section \ref{sec_particular_cases} we analyze some particular cases of the main theorem trying to emphasize its real scope. In Section \ref{sec_proof} the main result is finally proved. $\underset{\;}{\underset{\;}{\;}}$

\noindent\textbf{Remark.}
Although we used the example of gravity emerging from noncommutativity as a motivating context, we would like to emphasize that this paper is not intended to focus on it or in building concrete examples. Indeed, our aim is to propose a formalization to the notion of strong emergence, at least in the context of Lagrangian field theories, and a general strategy to investigate the existence problem for such phenomena. With this remark we admit the need of additional concrete applications of the methods proposed here, which should appear in a future work. In particular, the case of gravity emerging from noncommutative in its difference incarnations, is to appear in a work in progress.

\section{The Strong Emergence Problem}\label{sec_1}
\quad \;\,Recall that a \emph{field theory} on a $n$-dimensional manifold $M$ (regarded as the spacetime)
is given by an action functional $S[\varphi]$, defined in some space of fields (or configuration space) $\operatorname{Fields}(M)$, typically the space of sections of some real or complex vector bundle $E\rightarrow M$,
the \emph{field bundle}. A \emph{parameterized field theory} consists
of another bundle $P\rightarrow M$ (the \emph{parameter bundle}),
a subset $\operatorname{Par}(P)\subset\Gamma(P)$ of global sections
(the \emph{parameters})\emph{ }and a collection $S_{\varepsilon}[\varphi]$
of field theories, one for each parameter $\varepsilon\in\operatorname{Par}(P)$.
A more suggestive notation should be $S[\varphi;\varepsilon]$. So, e.g, for
the trivial parameter bundle $P\simeq M\times\mathbb{K}$ we have
$\Gamma(P)\simeq C^{\infty}(M;\mathbb{K})$ and in this case we say
that we have \emph{scalar parameters.} If we consider only scalar
parameters which are constant functions, then a parameterized theory
becomes the same thing as a 1-parameter family of field theories.
Here, and throughout the paper, $\mathbb{K}=\mathbb{R}$ or $\mathbb{K}=\mathbb{C}$
depending on whether the field bundle in consideration is real or
complex.

We will think of a parameter $\varepsilon$ as some kind of ``physical
scale'', so that for two given parameters $\varepsilon$ and $\varepsilon'$,
we regard $S[\varphi;\varepsilon]$ and $S[\varphi;\varepsilon']$
as \emph{the same theory} in two different physical scales. Notice
that if $P$ has rank $l$, then we can locally write $\varepsilon=\sum\varepsilon^{i}e_{i}$,
with $i=1,...,l$, where $e_{i}$ is a local basis for $\Gamma(P)$.
Thus, locally each physical scale is completely determined by $l$
scalar parameters $\varepsilon^{i}$ which are the fundamental ones.
In terms of these definitions, Question \ref{question_int} has a
natural generalization:
\begin{question}
\label{problem_1}Let $S_{1}[\varphi;\varepsilon]$ and $S_{2}[\psi;\delta]$
be two parameterized theories defined on the same spacetime $M$,
but possibly with different field bundles $E_{1}$ and $E_{2}$, and
different parameter bundles $P_{1}$ and $P_{2}$. Arbitrarily giving
a field $\varphi\in\Gamma(E_{1})$ and a parameter $\varepsilon\in\operatorname{Par}_{1}(P_{1})$,
can we find some field $\psi(\varphi)\in\Gamma(E_{2})$ and some parameter
$\delta(\varepsilon)\in\operatorname{Par}_{2}(P_{2})$ such that $S_{1}[\varphi;\varepsilon]=S_{2}[\psi(\varphi);\delta(\varepsilon)]$?
In more concise terms, are there functions $F:\operatorname{Par}_{1}(P_{1})\rightarrow\operatorname{Par}_{2}(P_{2})$
and $G:\Gamma(E_{1})\rightarrow\Gamma(E_{2})$ such that $S_{1}[\varphi;\varepsilon]=S_{2}[G(\varphi);F(\varepsilon)]$?
\end{question}
We say that the theory $S_{1}[\varphi;\varepsilon]$ \emph{emerges}
\emph{from the theory }$S_{2}[\psi;\delta]$ if the problem above
has a positive solution, i.e, if we can fully describe $S_{1}$ in
terms of $S_{2}$. Notice, however, that as stated the emergence problem
is fairly general. Indeed, if $P_{1}$ and $P_{2}$ have different
ranks, then, by the previous discussion, this means that the parameterized
theories $S_{1}$ and $S_{2}$ have a different number of fundamental
scales, so that we should not expect an emergence relation between
them. This leads us to think of considering only the case in which
$P_{1}=P_{2}$. However, we could also consider the situations in
which $P_{1}\neq P_{2}$, but $P_{2}=f(P_{1})$ is some nice function
of $P_{1}$, e.g, $P_{2}=P_{1}\times P_{1}\times...\times P_{1}$.
In these cases the fundamental scales remain only those of $P_{1}$,
since from them we can generate those in the product. Throughout this
paper we will also work with different theories defined on the same
fields, i.e, $E_{1}=E_{2}$. This will allow us to search for emergence
relations in which $G$ is the identity map $G(\varphi)=\varphi$. 

Hence, after these hypotheses, we can rewrite our main problem, whose
affirmative solutions axiomatize the notion of strong emergence we
are searching for:
\begin{question}
\label{problem_2}Let $S_{1}[\varphi;\varepsilon]$ and $S_{2}[\varphi;\delta]$
be two parametrized theories defined on the same spacetime $M$, on
the same field bundle $E$ and on the parameter bundles $P_{1}$ and
$P_{2}=f(P_{1})$, respectively. Does there exists some map $F:\operatorname{Par}_{1}(P_{1})\rightarrow\operatorname{Par}_{2}(f(P_{1}))$
such that $S_{1}[\varphi;\varepsilon]=S_{2}[\varphi,F(\varepsilon)]$?
\end{question}

Our plan is to show that the problem in Question \ref{problem_2}
has an affirmative solution for certain class of parameterized field theories, which by the absence of a better name we call \textit{generalized parameterized field theories}. This will be obtained from the class of \textit{differential parameterized theories} by means of allowing dynamical operators which are not necessarily differential, but some generalized version of them. 

\subsection{Differential Parameterized Theories}
\quad\;\,In order to motivate the need for looking at a special class of field theories, we begin by
noticing that typically the field theories are \textit{local}, in the sense that they are defined by means of integrating some Lagrangian density $\mathcal{L}(x,\varphi,\partial \varphi,\partial^2\varphi,...)$, i.e, $S[\varphi]=\int_M \mathcal{L}(j^{\infty}\varphi)dx^n$, where $j^{\infty}\varphi=(x,\varphi,\partial \varphi,\partial^2\varphi,...)$ is the jet prolongation. On the other hand, a quick look at the standard examples of field theories shows that, when working in a spacetime without boundary, after
integration by parts and using Stoke's theorem, those 
field theories can be stated, at least locally, in the form $\mathcal{L}(x,\varphi,\partial\varphi)=\langle\varphi,D\varphi\rangle$,
where $\langle\varphi,\varphi'\rangle$ is a nondegenerate pairing
on the space of fields $\Gamma(E)$ and $D:\Gamma(E)\rightarrow\Gamma(E)$
is a differential operator of degree $d$, which means that it can
be locally written as $D\varphi(x)=\sum_{\vert\alpha\vert\leq d}a_{\alpha}(x)\partial^{\alpha}\varphi$,
where $\alpha=(\alpha_{1},...,\alpha_{r})$ is some mult-index, $\vert\alpha\vert=\alpha_{1}+...+\alpha_{r}$
is its degree and $\partial^{\alpha}=\partial_{1}^{\alpha_{1}}\circ...\circ\partial_{r}^{\alpha_{r}}$,
with $\partial_{i}^{l}=\partial^{l}/\partial^{l}x_{i}$. This is
the case, e.g, of $\varphi^{3}$ and $\varphi^{4}$ scalar field theories,
the standard spinorial field theories, Einstein-Hilbert-Palatini-Type theories \cite{eu_EHP_1,eu_EHP_2} as well as Yang-Mills-type theories and certain canonical extensions of them \cite{eu_YMT_1,eu_YMT_2}.
More generally, recall that the first step in building the Feynman
rules of a field theory is to find the (kinematic part of the) operator
$D$ and take its ``propagator'' .

In these examples, the pairing $\langle\varphi,\varphi'\rangle$ is typically symmetric
(resp. skew-symmetric) and the operator $D$ is formally self-adjoint
(resp. formally anti-self-adjoint) relative to that pairing. Furthermore,
$\langle\varphi,\varphi'\rangle$ is usually a $L^{2}$-pairing induced
by a semi-Riemannian metric $g$ on the field bundle $E$ and/or on
the spacetime $M$, while $D$ is usually a generalized Laplacian
or a Dirac-type operator relative to $g$ \citep{costello2011renormalization}.
For example, this holds for the concrete field theories (scalar, spinorial
and Yang-Mills) above. The skew-symmetric case generally arises in
gauge theories (BV-BRST quantization) after introducing the Faddeev-Popov
ghosts/anti-ghosts and it depends on the grading introduced by the
ghost number \citep{costello2011renormalization}.

Another remark, still concerning the concrete situations above,
is that if the metric $g$ inducing the pairing $\langle\varphi,\varphi'\rangle$
is actually Riemannian (which means that the gravitational background
is Euclidean), then $\langle\varphi,\varphi'\rangle$ becomes a genuine
$L^{2}$-inner product and $D$ is elliptic and extends to a bounded
self-adjoint operator between Sobolev spaces \citep{donaldson1990geometry,freed2012instantons}.
Working with elliptic operators is very useful, since they always
admits parametrices (which in this Euclidean cases are the propagators)
and for generalized Laplacians the heat kernel not only exists, but
also has a well-known asymptotic behavior \citep{heat_kernel_generalized_laplacian},
which is very nice in the Dirac-type case \citep{berline2003heat}.

From the discussion above, it is natural to focus on parameterized
theories such that each functional $S[\varphi;\varepsilon]$ is local and defined by a Lagrangian density of the
$\mathcal{L}(\jmath^{\infty}\varphi
)=\langle\varphi,D_{\varepsilon}\varphi\rangle$,
i.e, are determined by a single nondegenerate pairing $\langle\varphi,\varphi'\rangle$
in $\Gamma(E)$, fixed a priori by the nature of $M$ and $E$, and
by a family of differential operators $D_{\varepsilon}\in\operatorname{Diff}(E)$,
one for each parameter $\varepsilon\in\operatorname{Par}(P)$, where
$\operatorname{Diff}(E)=\bigoplus_{d}\operatorname{Diff}^{d}(E;E)$
denotes the $\mathbb{K}$-vector space of all differential operators in $E$ (which is actually a $\mathbb{K}$-algebra which the composition operation) and $\operatorname{Diff}^{d}(E;E)$ is the space of those operators of degree $d$. Thus:

\begin{definition}\label{definition_background}
Let $M$ be a compact and oriented manifold.  A \textit{background for doing emergence theory} (or simply \textit{background}) over $M$ is given by the following data:
\begin{enumerate}
\item a $\mathbb{K}$-vector bundle $E\rightarrow M$ (the field bundle);
\item a pairing $\langle\varphi,\varphi'\rangle$ in $\Gamma(E)$;
\item a parameter bundle $P\rightarrow M$ and a set of parameters $\operatorname{Par}(P)\subset\Gamma(P).$
\end{enumerate}
\end{definition}
\begin{definition}\label{definition_2}
A \textit{differential parameterized theory} in a background over $M$ is a collection of differential operators $D_{\varepsilon}\in\operatorname{Diff}(E)$ with $\varepsilon \in \operatorname{Par}(P)$. The \textit{parameterized Lagrangian density} is given by $\mathcal{L}(\varphi;\varepsilon)=\langle \varphi,D_{\varepsilon}\varphi \rangle$. The \textit{parameterized action functional} is the integral of the parameterized Lagrangian in $M$.
\end{definition}

\subsection{Generalized Parameterized Theories}\label{sec_GPT}

\quad \;\,Sometimes working on the narrow class of field theories defined by differential operators is not enough. For instance, notice that in building the ``propagator'' of a Lagrangian field theory $\mathcal{L}(\jmath^{\infty}\varphi)=\langle \varphi, D\varphi \rangle$ we are actually
finding some kind of ``quasi-inverse'' $D^{-1}$ for the differential operator $D$. For example, in the Riemannian case, where differential operators $D\in\operatorname{Diff}^d(E)$ extend to bounded operators $\hat{D}\in B(W^{d,2}(E))$ in the Sobolev space, if $D$ is elliptic, then building the propagator is equivalent to building parametrices, which in turn can described in terms of Fredholm inverses for $\hat{D}$ \cite{costello2011renormalization}. On the other hand, in the Lorentzian setting (where the spacetime manifold is assumed globally hyperbolic and the typical differential operators are hyperbolic), building the propagator is about finding its advanced and retarded Green functions \cite{hyperbolic_1,hyperbolic_2}.   

Independently of the case, it would be very useful if the quasi-inverse $D^{-1}$ could exist
as a differential operator, i.e., if $D^{-1}\in\operatorname{Diff}(E)$. Indeed, in this case $D^{-1}$ would also define a differential field theory. In particular, in the parameterized case, if each $D_{\varepsilon}$ has a ``quasi-inverse'' described by a differential operator $D^{-1}_{\varepsilon}$, the collection of them define a parameterized differential theory over the same background. In turn, from the physical viewpoint, it would be very interesting if the ``quasi-inverse'' $D^{-1}$ of the differential operator $D$ was a genuine inverse for $D$ in the algebra $\operatorname{Diff}(E)$. Indeed, in this case one could use the relations $D\circ D^{-1}=I$ and $D^{-1}\circ D=I$  in order to get global solutions for the equation of motion $D\varphi=0$ \cite{taylor_II}. 

Notice, however, that the equations $D\circ D^{-1}=I=D^{-1}\circ D$ typically has no solutions in $\operatorname{Diff}(E)$. Indeed, recall that $\operatorname{Diff}(E)=\bigoplus_{d}\operatorname{Diff}^{d}(E;E)$ is a $\mathbb{Z}_{\geq0}$-graded algebra with composition, so that if $D^{-1}$ was a left or right inverse for $D$, then $\deg(D^{-1})=-\deg D$, which has no solution if $\deg D\neq 0$. This forces us to search for extensions of the algebra $\operatorname{Diff}(E)$ such that left and/or right inverses of differential operators may exist. The obvious idea is to consider extensions by adding a negative grading to $\operatorname{Diff}(E)$ getting the structure of a $\mathbb{Z}$-graded $\mathbb{K}$-algebra, and the natural choice is the $\mathbb{Z}$-graded algebra $\operatorname{Psd}(E)$ of pseudo-differential operators, which are defined via symbol-theoretic (i.e., microlocal analysis) approach \cite{algebra_pseudo_0,algebra_pseudo_1}. 

For constant coefficient operators, the right-inverse really exist in $\operatorname{Psd}(E)$. In the case of non-constant coefficients, for Riemannian spacetimes, elliptic operators satisfying specific ellipticity conditions admit right-inverse in $\operatorname{Psd}(E)$  \citep{right_inverse_const_1,right_inverse_const_2,right_inv_1,right_inv_1_2,right_inv_2,right_inv_3,right_inv_4,ghaemi2016study}. For globally hyperbolic Lorentzian spacetimes, at least when restricted to fields with support in the causal cones, Green hyperbolic operators (in particular normally hyperbolic operators) have both left and right inverses in $\operatorname{Psd}(E)$ \cite{hyperbolic_1,hyperbolic_2}. For flat spacetimes, more abstract examples exist \cite{algebra_pseudo_2}. On the other hand, if a differential operator is not right-invertible in $\operatorname{Psd}(E)$ it could be invertible in another extension of $\operatorname{Psd}(E)$, such as in the class of Fourier-type operators \cite{taylor_II}.

Our main result does not depend on the extension of $\operatorname{Diff}(E)$; the only thing we need is the existence of a right-inverse for a differential operator as \textit{some} kind of operator, i.e., we only need some $\mathbb{K}$-algebra $\operatorname{Op}(E)$ such that $\operatorname{Diff}(E)\subset \operatorname{Op}(E)\subset\operatorname{End}(\Gamma(E))$ and such that the subset of differential operators $\operatorname{Op}(E)$ contains nontrivial elements. 

\begin{definition}\label{definition_GB}
Let $M$ be a compact and orientable manifold. A \textit{generalized background} over $M$, denoted by $\operatorname{GB}(M)$ is given by the same data in Definition \ref{definition_background} and, in addition, a $\mathbb{K}$-algebra $\operatorname{Op}(E)\subset\operatorname{End}(\Gamma(E))$ extending $\operatorname{Diff}(E)$ and such that $\operatorname{Diff}(E)\cap R\operatorname{Op}(E)$ contains non-multiplies of the identity, where  $R\operatorname{Op}(E)$ is the set of right-invertible elements of $\operatorname{Op}(E)$.

\end{definition}
\begin{definition}
A \textit{generalized parameterized theory} (GPT) in a generalized background $\operatorname{GB}(M)$ is a collection of generalized operators $\Psi _{\varepsilon}\in \operatorname{Op}(E)$, with $\varepsilon \in \operatorname{Par}(P)$. The corresponding Lagrangian density action funcional are defined analogously to Definition \ref{definition_2}, just replacing $D_\varepsilon$ with $\Psi_{\varepsilon}$.
\end{definition}
\begin{definition}\label{definition_GPT}
We say that a GPT over $M$ is \textit{right-invertible} if the generalized operators $\Psi_{\varepsilon}$ are right-invertible, i.e, if they belongs to $R\operatorname{Op}(E)$ for every $\varepsilon \in \operatorname{Par}(P)$.
\end{definition}

\subsection{Polynomial Parameterized Theories}

\quad\;\, Until this moment we considered theories which are defined by a 1-parameter family of differential (or more general) operators. As discussed, these provide a description of \textit{free} theories and of their parameterized version. We will show, however, that the concept in broad enough to  describe interacting theories as well. In order to do this, recall that the classical examples of interacting field theories has interacting terms given by multivariate polynomials with variables corresponding to operators of different theories at interaction, as discussed at the introductory section.

But, in order to talk about polynomials we need a ring of coefficients. In our parameterized context, the natural idea is to consider coefficients depending on the parameters. This leads us to look at the ring  $\operatorname{Map}(\operatorname{Par}(P);\mathbb{K})$ of scalar functions $f:\operatorname{Par}(P)\rightarrow \mathbb{K}$ on the space of parameters. For a given $l\geq0$ and formal variables $x_1,...,x_r$, let $\operatorname{Map}_l(\operatorname{Par}(P);\mathbb{K})[x_1,...,x_r]$ be the corresponding $\mathbb{K}$-vector space of multivariate polynomials of degree $l$ in the given variables. Thus, an element of it can be  written as $p^{l}[x_{1},...,x_{r}]=\sum_{\vert\alpha\vert\leq l}f_{\alpha}\cdot x^{\alpha}$, where  $f_{\alpha}\in\operatorname{Map}(\operatorname{Par}(P);\mathbb{K})$ and $\alpha$ is a multi-index. If $\Psi_{1},...,\Psi_{r}\in\operatorname{Op}(E)$
are fixed generalized operators and $p^{l}[x_{1},...,x_{r}]$ is a polynomial
as above, recalling that $\operatorname{Op}(E)$ is a $\mathbb{K}$-algebra, by means of replacing the formal variables $x_{i}$
with the operators $\Psi_{i}$ we get a family of operators $p^{l}[\Psi_{1},...,\Psi_{r}](\varepsilon)=\sum_{\vert\alpha\vert\leq l}f_{\alpha}(\varepsilon)\Psi^{\alpha}$, with $\varepsilon\in \operatorname{Par}(P)$. The following definitions are then natural.

\begin{definition}\label{definition_PPT}
Let $M$ be a compact and oriented manifold. A \textit{polynomial parameterized theory \emph{(}PPT\emph{)} of degree $l\geq0$ in $r$ variables}, defined on a generalized context $\operatorname{GB}(M)$, is given by an element $p^l$ of $\operatorname{Map}_l(\operatorname{Par}(P);\mathbb{K})[x_1,...,x_r]$ and by operators $\Psi_i \in \operatorname{Op}(E)$, with $i=1,...,r$. The parameterized Lagrangian and the parameterized action functional are $$\mathcal{L}(\jmath^{\infty}\varphi;\varepsilon)=\langle \varphi, p^{l}[\Psi_{1},...,\Psi_{r}](\varepsilon)\varphi \rangle \quad \text{and} \quad S[\varphi;\varepsilon]=\ll \varphi, p^{l}[\Psi_{1},...,\Psi_{r}](\varepsilon)\varphi \gg,$$
while the extended Lagrangian and the extended action functional are defined analogously by means of replacing $\psi_i$ with $\widetilde{\Psi_i}$.
\end{definition}

Now, to check that the definition above really captures the standard examples of interacting terms, if the interaction to be described is of $j>1$ different fields, the corresponding PPT is usually of $j$ variables and such that $E\simeq \oplus_i E_i$, with $i=1,...,j$, and $\Psi_i=(0,...,\Phi_i,...,0)$, where $\Psi_i\in \Gamma (E)$ and $\Phi_i\in \Gamma(E_i)$. Here, $E_i$ is the field bundle of the $i$th field theory\footnote{Gauge field with gauge group $G$ are incorporated as follows. Recall that they are principal $G$-connections, which can be regarded as 1-forms in $M$ taking values in the adjoint $G$-bundle $E_{\mathfrak{g}}$. Thus, we just take $E_i=TM^*\otimes E_{\mathfrak{g}}$.}. 

\begin{rem}
Every PPT of arbitrary degree $l$ and in arbitrary number $r$ of variables is a GPT, on the same generalized background, with parameterized operator $\Psi_\varepsilon = p^{l}[\Psi_{1},...,\Psi_{r}](\varepsilon)$, so that the concept of GPT is really broad enough to describe both free and interacting terms of a typical Lagrangian field theory.
\end{rem}

Closing this section, notice that a priori we have two possible definitions of right-invertibility for a PPT. Since by the above every PPT is a GPT, we could say that a PPT is \textit{right-invertible} if it is as a GPT, i.e., if for every $\varepsilon$ the generalized operator $p^{l}[\Psi_{1},...,\Psi_{r}](\varepsilon)$ is right-invertible. But we could also define a \textit{right-invertible} PPT as such that each generalized operator $\Psi_i$, with $i=1,...,r$, is right-invertible. These conditions are very different. Indeed, the first one is about the invertibility of \textit{multivariate polynomials} in noncommutative variables, while the second one is about the invertibility of the \textit{variables} of a multivariate polynomial. Thus, the first one relies on constraints on the polynomials $p^l$, while the second one is a condition only on the operators $\Psi_i$. Luckly, we will need only the second condition, leading us to the define:

\begin{definition}
We say that a PPT in $r$ variables is \textit{right-invertible} if the defining operators $\Psi_i$, with $i=1,...,r$, are right-invertible.
\end{definition}

\subsection{Higher Parameter Degree}
\quad\;\,Notice that, as in Definitions \ref{definition_2}, \ref{definition_GPT} and \ref{definition_PPT}, a parameterized field theory has as parameters a \textit{single} element $\varepsilon \in \operatorname{Par}(P)$.  On the other hand, as discussed at the beginning of Section \ref{sec_1}, a physical theory may depend on many fundamental scales. This leads us to consider parameterized theories depending on a list $\varepsilon(\ell)=(\varepsilon_1,...,\varepsilon_\ell)$ of elements of $\operatorname{Par}(P)$, i.e, such that $\varepsilon(\ell)\in \operatorname{Par}(P)^{\ell}$. In this case, we say that the number $\ell\geq1$ is the \textit{fundamental degree} of the parameterized theory. We will also use
the convention that $\operatorname{Par}(P)^{0}$ is a singleton, whose
element we denote by $\varepsilon(0)$. More precisely, we have the following definition.

\begin{definition}
Let $M$ be a compact and oriented manifold. A \textit{GPT of degree $\ell\geq 0$} in a generalized background $\operatorname{GB}(M)$ is given by a collection of generalized operators $\Psi_{\varepsilon(\ell)}\in \operatorname{Op}(E)$. Particularly, a \textit{PPT of degree $(l,\ell)$ in $r$ variables} is given by an element $p^l_{\ell}\in \operatorname{Map}_l(\operatorname{Par}(P)^{\ell};\mathbb{K})[x_1,...,x_r]$ and by operators $\Psi_i\in \operatorname{Op}(E)$, with $i=1,...,r$.
\end{definition}

\section{The Emergence Theorem}\label{sec_emergence_thm}
\quad\;\,After the previous digression we are now ready to state our main theorem. The syntax is the following:$\underset{\;}{\underset{\;}{\;}}$

\noindent \textbf{Main Theorem Syntax.} \textit{Let $M$ be a compact and oriented manifold with a fixed generalized background $\operatorname{GB}(M)$. Let $S_1$ be a GPT of degree $\ell$ and $S_2$ a PPT of degree $(l,\ell')$ in $r$ variables. Suppose that:
    \begin{enumerate}
        \item $S_1$ depends on the fundamental parameters $\varepsilon(\ell)$ in a ``linear'' way;
        \item the PPT theory $S_2$ is right-invertible and its coefficients $f_\alpha$ are ``suitable'' functions in a sense that depends on $r,l$ and $\ell'$.
    \end{enumerate}
Then $S_1$ emerges from} $S_2$.$\underset{\;}{\underset{\;}{\;}}$

In order to turn this syntax into a rigorous statement, let us described what we meant by a ``linear'' dependence on the parameters and by ``suitable'' functions. Since the term ``linear'' typically means ``preservation of some algebraic structure'', it is implicit that we will need to assume some algebraic structure on the space of parameters $\operatorname{Par}(P)^{\ell}$. To begin, we will require an associative and unital $\mathbb{K}$-algebra structure, whose sum and multiplication we will denote by  ``$+^{\ell}$'' and ``$*^{\ell}$'', respectively, or simply by ``$+$'' and ``$*$'' when the number $\ell$ of fundamental parameters is implicit. Thus,

\begin{definition}
A GPT of degree $\ell$ in a generalized background $\operatorname{GB}(M)$ is  \textit{linear} (or \textit{homomorphic}) if $\operatorname{Par}(P)^{\ell}$ has a $\mathbb{K}$-algebra structure and the rule $\varepsilon(\ell)\mapsto \Psi_{\varepsilon(\ell)}$ is a $\mathbb{K}$-algebra homomorphism.
\end{definition}

Some properties of the emergence phenomena (as those proved in Subsection \ref{sec_properties}) depends only on the preservation of sum ``$+$'', scalar multiplication, or product ``$*$''. This motivates the following definition:

\begin{definition}
Under the same notations and hypotheses of the last definition, we say that a GPT is \textit{additive} (resp. \textit{multiplicative}) if the rule $\varepsilon(\ell)\mapsto \Psi_{\varepsilon(\ell)}$ is $\mathbb{K}$-linear (resp. a multiplicative monoid homomorphism). If it is only required $\Psi_{c\varepsilon(\ell)}=c\Psi_{\varepsilon(\ell)}$ we say that the GPT is \textit{scalar invariant}.
\end{definition}

On the other hand, notice that in a PPT the generalized operators $\Psi$ always appear multiplied by a \textit{number} $f(\varepsilon(\ell))\in \mathbb{K}$, which depends on the parameters $\varepsilon(\ell)$. However, if one recalls that the parameters are interpreted as fundamental scales, it should be natural to consider parameters $\varepsilon(\ell)$ (instead of numbers $f(\varepsilon(\ell))\in \mathbb{K}$ assigned to them)  multiplying the operators $\Psi$. Thus, we need an action $
\cdot^{\ell}:\operatorname{Par}(P)^{\ell}\times\operatorname{Op}(E)\rightarrow\operatorname{Op}(E)$. But, since by Definition \ref{definition_GB} and by the above assumption both $\operatorname{Op}(E)$ and $\operatorname{Par}(P)^{\ell}$ are $\mathbb{K}$-algebras, it is natural to require some compatibility between the action $\cdot^{\ell}$ and these $\mathbb{K}$-algebra structures. More precisely, we will assume that $\cdot^{\ell}$ is $\mathbb{K}$-bilinear and that the following condition are satisfied for every $\Psi,\Psi'\in\operatorname{Op}(E)$ and every $\varepsilon(\ell)\in \operatorname{Par}(P)^{\ell}$:
\begin{equation}
(\varepsilon(\ell)\cdot^{\ell}\Psi)\circ\Psi'=\varepsilon(\ell)\cdot^{\ell}(\Psi\circ\Psi') \quad \text{and} \quad (\varepsilon(\ell)*\delta(\ell))\cdot^{\ell} \Psi =  \varepsilon(\ell)\cdot^{\ell}(\delta(\ell)\cdot^{\ell}\Psi). \label{compatibility_action}
\end{equation}

\begin{rem}
If not only conditions (\ref{compatibility_action}) are satisfied, but also 
\begin{equation}\label{compatibility_2}
   \Psi\circ(\varepsilon(\ell)\cdot^{\ell} \Psi')=\varepsilon(\ell) \cdot^{\ell} (\Psi\circ \Psi') 
\end{equation}
is satisfied for every $\varepsilon(\ell)$ and every $\Psi,\Psi'$, then the algebra $\operatorname{Par}(P)^{\ell}$ must be commutative. See Comment \ref{comment_3}. On the other hand, if (\ref{compatibility_action}) is satisfied for every $\varepsilon(\ell)$ and every $\Psi,\Psi'$, but (\ref{compatibility_2}) is satisfied for every $\varepsilon(\ell)$ and a \textit{single} $ \Psi' = \Psi$, then $\operatorname{Par}(P)^{\ell}$ need not be commutative. \end{rem}

Now, to make a rigorous sense of Condition 2 in the previous syntactic statement, let us clarify what we mean by a ``suitable'' function $f \in\operatorname{Map}(\operatorname{Par}(P)^{\ell};\mathbb{K})$. In few words, a function $f$ is ``suitable'' if it belongs to some functional calculus. For us, a \textit{functional calculus of degree $\ell$} is a subset $C_{\ell}(P;\mathbb{K})$ of $\operatorname{Map}(\operatorname{Par}(P)^{\ell};\mathbb{K})$ endowed with a function $\Psi^{\ell}_{-}:C_{\ell}(P;\mathbb{K})\rightarrow \operatorname{Op}(E)$ assigning to each function $f\in C_{\ell}(P;\mathbb{K})$ a generalized operator $\Psi_f^{\ell}$, which is compatible with the action $\cdot^{\ell}$, in the sense that
\begin{equation}
\Psi_{f}^{\ell}\circ[f(\varepsilon(\ell))\Psi]=\varepsilon(\ell)\cdot^{\ell}\Psi.\label{functional_calculus}
\end{equation}

\begin{definition}
A functional calculus is \textit{unital} if $C_\ell(P;\mathbb{K})$ contains the constant function $f\equiv 1$. 
\end{definition}

We notice that, although our main emergence theorem will depend on the choice of a unital functional calculus, the existence of them is not an obstruction.

\begin{lem}\label{lemma_functional_calculus}
In every generalized background $\operatorname{GB}(M)$, for every $\ell\geq 0$ there exists a unique functional calculus of degree $\ell$ such that $C_\ell(P;\mathbb{K})$ is the set of nowhere vanishing functions $f:\operatorname{Par}(P)^{\ell}\rightarrow \mathbb{K}$ if $\ell >0$, and the constant function $f\equiv 1$ if $\ell =0$.
\end{lem}

\begin{proof}
First, assume existence for every $\ell \geq 0$. Uniqueness in the case $\ell =0$ is obvious. In the case $\ell >0$, from the existence hypothesis we have $\Psi_{f}^{\ell}\circ[f(\varepsilon(\ell))\Psi]=\varepsilon(\ell)\cdot^{\ell}\Psi$
for every $f$, $\Psi$ and $\varepsilon(\ell)$, so that $\Psi_{f}^{\ell}\circ\Psi=(\varepsilon(\ell)\cdot^{\ell}\Psi)/f(\varepsilon(\ell))$.
In particular, for $\Psi=I$, we get $\Psi_{f}^{\ell}=(\varepsilon(\ell)\cdot^{\ell}I)/f(\varepsilon(\ell))$, proving uniqueness.
In order to prove existence, for each $\ell\geq 0$ define $\Psi_{f}^{\ell}=(\varepsilon(\ell)\cdot^{\ell}I)/f(\varepsilon(\ell))$,
so that 
\[
\Psi_{f}^{\ell}\circ[f(\varepsilon(\ell))\Psi]=(\varepsilon(\ell)\cdot^{\ell}I)\circ\Psi=\varepsilon(\ell)\cdot^{\ell}(I\circ\Psi)=\varepsilon(\ell)\cdot^{\ell}\Psi,
\]
where in the last step we used the compatibility between $\cdot^{\ell}$
and $\circ$, as described in (\ref{compatibility_action}).
\end{proof}

Finally, let us state a few more necessary technical assumptions:

\begin{enumerate}
    \item \textit{The $\mathbb{K}$-algebra of fundamental parameters $\operatorname{Par}(P)^{\ell}$ is required to have square roots.} This means that the function $\varepsilon(\ell)\mapsto \varepsilon(\ell)*\varepsilon(\ell)$ is surjective, i.e, for every $\varepsilon(\ell)$ we there is another $\sqrt{\varepsilon(\ell)}\in \operatorname{Par}(P)^{\ell}$ such that $$(\sqrt{\varepsilon(\ell)})^2=\sqrt{\varepsilon(\ell)}*\sqrt{\varepsilon(\ell)}=\varepsilon(\ell).$$
    
    \item \textit{Right multiplication by $I$ is injective}. More precisely, for every $\varepsilon(\ell),\delta(\ell)\in \operatorname{Par}(P)^{\ell}$, if $\varepsilon(\ell)\cdot^{\ell}I=\delta(\ell)\cdot ^{\ell}I$, then $\varepsilon(\ell)=\delta(\ell)$, i.e., $\varphi_i=\delta_i$ for $i=1,...,\ell$. \label{injective_action}
\end{enumerate}

$\quad \;\,$ Some comments concerning these technical conditions:
\begin{enumerate}
    \item The square roots $\sqrt{\varepsilon(\ell)}$ in condition 1. need not be unique.
    \item Condition \ref{injective_action} above and compatibility conditions (\ref{compatibility_action}) imply that the right multiplication $r^{\ell}_{\Psi}:\operatorname{Par}^{\ell}(P)\rightarrow \operatorname{Op}(E)$ is injective for every right-invertible operator $\Psi\in R\operatorname{Op}(E)$. Thus, for every such $\Psi$, the map $r^\ell_{\Psi}$ is actually an isomorphism 
$\operatorname{Par}(P)^{\ell}\simeq\operatorname{Par}(P)^{\ell}\cdot^{\ell}\Psi$
between its domain and its image.
\item If both compatibility conditions (\ref{compatibility_action}) and (\ref{compatibility_2}) are satisfied, then Condition \ref{injective_action} above implies that the $\mathbb{K}$-algebra $(\operatorname{Par}(P)^{\ell},*^{\ell})$ is commutative. This follows basically from the Eckmann-Hilton argument \cite{eckmann_1,eckmann_2}. \label{comment_3}
\end{enumerate}

Now, suppose that $\operatorname{Par}(P)^{\ell}$ has a $\mathbb{K}$-algebra structure. Then $\operatorname{Par}(P)^{k\ell}$ has an induced $\mathbb{K}$-algebra structure, with $k>0$, given by componentwise sum and multiplication. If $\operatorname{Par}(P)^{\ell}$ has square roots, then $\operatorname{Par}(P)^{k\ell}$ has too, given by $\sqrt{\varepsilon(k\ell)}=(\sqrt{\varepsilon(\ell_1)},...,\sqrt{\varepsilon(\ell_k)})$, where $\varepsilon(k\ell)=(\varepsilon(\ell_1),...,\varepsilon(\ell_k))$ and $\varepsilon(\ell_i)=(\varepsilon_{i,1},...,\varepsilon_{i,\ell})$.  On the other hand, since $\operatorname{Par}(P)^{l\ell}$ embeds in $\operatorname{Par}(P)^{k\ell}$ as $\varepsilon(\ell)\mapsto (\varepsilon(\ell_1),...,\varepsilon(\ell_l),0,...,0)$, if $l\leq k$, it follows that every action $\cdot^{k\ell}$ of $\operatorname{Par}(P)^{k\ell}$ can be pulled back to an action $\cdot^{l\ell}$ of $\operatorname{Par}(P)^{l\ell}$. 

If  $C_{k\ell}(P;\mathbb{K})\subset \operatorname{Map}(\operatorname{Par}(P)^{k\ell};\mathbb{K})$ is a set of functions, pullback by the inclusion $\imath :\operatorname{Par}(P)^{l\ell} \hookrightarrow \operatorname{Par}(P)^{k\ell}$ above defines a new set of functions $C_{l\ell;k}(P;\mathbb{K})\subset \operatorname{Map}(\operatorname{Par}(P)^{l\ell};\mathbb{K})$. Furthermore, if $C_{k\ell}(P;\mathbb{K})$ is actually a functional calculus defined by a map $\Psi^{k\ell}_{-}:C_{k\ell}(P;\mathbb{K})\rightarrow \operatorname{Op}(E)$, we have an induced functional calculus in $C_{l\ell;k}(P;\mathbb{K})$, defined by the function $\Psi^{l\ell;k}_{-}:C_{l\ell;k}(P;\mathbb{K})\rightarrow \operatorname{Op}(E)$ such that $\Psi^{l\ell;k}_{\imath \circ f}=\Psi^{k\ell}_f$. Notice that if $C_{k\ell}(P;\mathbb{K})$ is unital and/or satisfies the third technical conditions above, then $C_{l\ell;k}(P;\mathbb{K})$ does.

\subsection{Formal Statement}\label{sec_formal_statement}

$\quad\;\,$ We can now rigorously state the main result of this paper. First, notice that the three technical conditions of last section are about algebraic properties of the space $\operatorname{Par}^{\ell}(P)$ of fundamental parameters and on its action on the algebra $\operatorname{Op}(E)$ of generalized operators. On the other hand, recall from Definition \ref{definition_GB} that both $\operatorname{Par}^{\ell}(P)$ and $\operatorname{Op}(E)$ are part of the data defining a generalized background. Thus, those conditions  are actually conditions on the underlying generalized background $\operatorname{GB}(M)$. This motivates the following definition.
\begin{definition}\label{def_GB_type}
Let $M$ be a compact and oriented smooth manifold and let $\ell\geq0$ and $k>0$ be non-negative integers. We say that a generalized background $\operatorname{GB}(M)$ is of \textit{$(\ell,k)$-type} if:
\begin{enumerate}
    \item the space of fundamental parameters $\operatorname{Par}^{\ell}(P)$ has an structure of $\mathbb{K}$-algebra with square roots;
    \item the induced $\mathbb{K}$-algebra   $\operatorname{Par}^{k\ell}(P)$ acts in $\operatorname{Op}(E)$ by an action $\cdot^{k\ell}$ which is injective at the identity operator $I$ and compatible with a functional calculus $C_{k\ell}(P;\mathbb{K})$.
\end{enumerate}
\end{definition}

Our main theorem is then the following:

\begin{thm}[Emergence Theorem]\label{main_theorem}
Let $M$ be a compact and oriented manifold and let $\operatorname{GB}(M)$ be a generalized background of $(\ell,k)$-type, with $k>0$. Let $S_1$ be a GPT of degree $\ell$ and let $S_2$ be a PPT of degree $(l,\ell')$ in $r$ variables, where $\ell '=k'\ell$, with $0<k'\leq k$, defined on $\operatorname{GB}(M)$. Assume:
\begin{enumerate}
    \item $S_1$ is homomorphic;
    \item $S_2$ is right-invertible and the coefficient functions $f_{\alpha}:\operatorname{Par}(P)^{k'\ell}\rightarrow \mathbb{K}$ belongs to the functional calculus $C_{k'\ell;k}(P;\mathbb{K})$. 
\end{enumerate}
Then $S_1$ emerges from $S_2$.
\end{thm}

Before proving the emergence theorem, let us say that the proof which will be given here can be generalized, with basically the same steps and arguments, in some directions (for more details, see \cite{eu_emergence_v1}):

\begin{enumerate}
    \item \textit{the spacetime manifold $M$ need not be compact nor orientable}. Notice that these conditions were used only to define integrals. Thus, instead, one could only assume integrability
conditions on global sections (such as compact supportness) and consider Lagrangians as taking values
on general densities. \label{generalization_1}
\item \textit{the space of fundamental parameters need not be a $\mathbb{K}$-algebra, but a more lax algebraic entity}. As will be clear during the proofs, we only need the ``nonnegative'' part of a $\mathbb{K}$-algebra. More precisely, what we really need is that the set $\operatorname{Par}(P)^\ell$ can be regarded as the subset of a $\mathbb{K}$-algebra $A$, which is closed by sum, multiplication, and scalar multiplication by $\mathbb{R}_{\geq0}$. Notice that if $P$ is a $\mathbb{K}$-algebra bundle, then $\operatorname{Par}(P)^\ell$ can always be realized as a subset of the $\mathbb{K}$-algebra $\Gamma(P)^\ell$.
\item the coefficient functions $f_\alpha$ need not be scalar functions, but actually maps $f_\alpha:\operatorname{Par}(P)^{\ell'}\rightarrow \operatorname{Par}(P)^{\ell}$, so that the scalar multiplication $f_\alpha(\varepsilon(\ell))\Psi$ is replaced by the action $\cdot^{\ell}$. The notions of functional calculus, etc., can be defined in an analogous way such that the syntax of the theorem remains the same\footnote{Recall from the discussion at Introduction that parameter-valued coefficient functions plays an important role in the description of emergent gravity in terms of emergence phenomena. Further explanations will appear in a work in progress.}.

\label{generalization_2}
  
\end{enumerate}

\subsection{Some Particular Cases}\label{sec_particular_cases}

\quad\;\,Let $M$ be a compact and orientable manifold, $\mathbb{K}=\mathbb{C}$ an $E=M\times \mathbb{C}$ the complex trivial line bundle, regarded as a field bundle with space of fields given by complex scalar functions $C^{\infty}(M;\mathbb{C})$, endowed with the pairing $\langle\varphi,\psi\rangle=\varphi \overline{\psi}$. In addition, let $P=M\times \mathbb{C}$, viewed now as the parameter bundle, and take the constant functions as parameters, so that $\operatorname{Par}(P)\simeq \mathbb{C}$. This data clearly defines a background over $M$. Let $\operatorname{Op}(E)$ be the complex algebra $\operatorname{Psd}(M\times\mathbb{C})$ of pseudo-differential operators. By the discussion of Section \ref{sec_GPT} we then have a generalized background. Notice that $\operatorname{Par}(P)\simeq \mathbb{C}$ is a $\mathbb{C}$-algebra with square roots and with their action in $\operatorname{Psd}(M\times\mathbb{C})$ via scalar multiplication, if $z\cdot I=z'\cdot I$, then clearly $z=z'$. Finally, let $C_{1}(P;\mathbb{C})$ be the unital functional calculus given by nowhere vanishing functions $f:\operatorname{Par}(P)\simeq \mathbb{C}\rightarrow \mathbb{C}$, as in Lemma \ref{lemma_functional_calculus}, defining a generalized background over $M$ of $(1,1)$-type.

As a particular case of Theorem \ref{main_theorem} we then have:

\begin{cor}
Let $M$ be a compact and oriented manifold and let $\operatorname{GB}(M)$ be the generalized background of $(1,1)$-type defined above. Let $D\in \operatorname{Diff}(M\times\mathbb{C})$ be any idempotent differential operator, i.e., there exists $n>0$ such that $D^{2 n}=D^n$. For given $r>0$ and $l\geq0$, let $p^l[x_1,..,x_r]$ be a polynomial of degree $l$ in $r$ variables and whose coefficients $f_{\alpha}:\mathbb{C}\rightarrow \mathbb{C}$ are nowhere vanishing functions. Let $D_1,...,D_r\in \operatorname{Diff}(M\times\mathbb{C})$ other differential operators and assume one of the following conditions:
\begin{enumerate}
    \item the operators $D_i$, with $i=1,...,r$ are of constant coefficient;
    \item there exists a Riemannian metric in $M$ such that $D_i$, with $i=1,...,r$ are strongly elliptic in the sense of any of references  \emph{\citep{right_inverse_const_1,right_inverse_const_2,right_inv_1,right_inv_1_2,right_inv_2,right_inv_3,right_inv_4,ghaemi2016study}};
    \item there exists a Lorentzian metric in $M$ such that $M$ is globally hyperbolic and each $D_i$, with $i=1,...,r$, is Green hyperbolic. 
\end{enumerate}
Then theory $\mathcal{L}_1 (\jmath^{\infty}\varphi;\varepsilon)=\overline{\varphi}\varepsilon D^n \varphi$ emerges from  theory $$\mathcal{L}_2(\jmath^{\infty}\varphi;\delta)=\overline{\varphi}p^l[D_1,...,D_r]\varphi=\sum_{\vert\alpha\vert\leq l} \overline{\varphi}f_{\alpha}(\delta)D^\alpha \varphi.$$
\end{cor}\label{particular_case}
\begin{proof}
Since $D^{2n}=D^n\circ D^n = D^n$, the rule $\varepsilon\mapsto \varepsilon D^n$ is clearly homomorphic. On the other hand, from the discussion on Section \ref{sec_GPT} each of the three hypotheses above implies that $D_i$, for $i=1,...,r$, are right-invertible as objects of $\operatorname{Psd}(M\times\mathbb{C})$. Since the coefficient functions $f_{\alpha}$ are nowhere vanishing and therefore belong to the functional calculus of $\operatorname{GB}(M)$, the result follows from Theorem \ref{main_theorem}.
\end{proof}

\begin{rem}\label{remark_particular_case}
From Comment \ref{generalization_1}, the same construction of $\operatorname{GB}(M)$ holds if $M$ is a bounded open set of some $\mathbb{R}^N$. From Comment \ref{generalization_2} it remains valid if $\operatorname{Par}(P)\simeq \mathbb{R}_{\geq0}$ are the constant non-negative real functions and the functional calculus $C_1(\mathbb{R}_{\geq 0}; \mathbb{C})$ consists of the nowhere vanishing functions taking values in $\mathbb{R}_{\geq 0}$. The difference, in this case, is that both $\varepsilon$ and $f_\alpha (\delta)$ are real numbers, so that the Lagrangians in the last corollary are real too. See \cite{eu_emergence_v1} for further details.
\end{rem}

Other generalizations of the last corollary, and particular cases of Theorem \ref{main_theorem}, are the following:  

\begin{enumerate}
\item \emph{We can consider other kind of fields}. In Corollary \ref{particular_case} we considered a generalized background defined on the trivial line bundle $M\times \mathbb{C}$. Notice, however, that if $E$ is any complex bundle with an Hermitian metric at the fibers, then the space of pseudo-differential operators $\operatorname{Psd}(E)$ remains well-defined  as a $\mathbb{Z}$-graded $\mathbb{C}$-algebra, so that the same thing holds equally well if instead of scalar fields one considered vector fields and tensor fields.  
\item \emph{We can consider other kind of parameters}. In Corollary \ref{particular_case}
we considered $\ell=1$ and $\operatorname{Par}(P)\simeq\mathbb{C}$ (or $\mathbb{R}_{\geq0}$, due to the remark above).
We could consider, more generally, $\operatorname{Par}(P)$ as any
complex algebra with square roots endowed with an action $\cdot :\operatorname{Par}(P)\times \operatorname{Psd}(E) \rightarrow \operatorname{Psd}(E)$, where
$E$ is a complex vector bundle (due to the last remark), such that conditions (\ref{compatibility_action}) and Condition $\ref{injective_action}$  are satisfied. Let $\Psi_0$ be such that $\Psi^n_0$ is idempotent and suppose that (\ref{compatibility_2}) are satisfied for fixed $\Psi=\Psi'=\Psi^n_0$. Then the rule $\varepsilon \mapsto \varepsilon\cdot\Psi^n_0$ is an algebra homomorphism and Corollary \ref{particular_case} holds equally well. 
\begin{example}[operator parameters]

Take $P=\operatorname{End}(E)$, so that $\Gamma(P)\simeq \operatorname{End}(\Gamma (E))$, which is an associative $\mathbb{C}$-algebra with an obvious action in $\operatorname{Psd}(E)$ by composition  such that conditions (\ref{compatibility_action}) and Condition \ref{injective_action} are clearly satisfied. Let $\Psi_0\in \operatorname{Psd}(E)$ be such that $\Psi^n_0$ is idempotent and  let $Z(\Psi^n_0)$ denote its centralizer, i.e., the subalgebra of all elements $\sigma\in \operatorname{End}(\Gamma(E))$ such that $\sigma \circ \Psi^n_0 = \Psi^n_0 \circ \sigma $, so that (\ref{compatibility_2}) is satisfied. Then take $\operatorname{Par}(P)$ as some subalgebra of $Z(\Psi^n_0)$ with square roots. As a concrete example, one can take  $\operatorname{Par}(P)$ as the subalgebra of nonnegative bounded self-adjoint operators in $\Gamma (E)$ which commutes with $\Psi^n_0$.
\end{example}
\item \emph{We can consider other kinds of parameterized operators. } In Corollary \ref{particular_case} and in the above generalizations we considered only parameterized operators of the form $\Psi_{\varepsilon}=\varepsilon \cdot \Psi$, which forced us to assume $\Psi$ idempotent. Indeed, notice that the nilpotency condition was used only to ensure that $\varepsilon \mapsto \varepsilon \cdot \Psi$ is an algebra homomorphism. More generally, let $\operatorname{Par}(P)$ be some complex algebra with square roots endowed with a representation $\rho:\operatorname{Par}(P) \rightarrow \operatorname{End}_{\mathbb{C}}(\Gamma (E))$ and define the action of $\operatorname{Par}(P)$ in $\operatorname{Psd}(E)$ by $\varepsilon\cdot \Psi := \rho(\varepsilon)\circ \Psi$, so that the second condition in (\ref{compatibility_action}) is clearly satisfied. If the action is faithful, then Condition \ref{injective_action} is satisfied too. Finally, if the action is compatible with the algebra structure of $\operatorname{Psd}(E)$, i.e., if  
\begin{equation}\label{boolean_condition}
   \rho(\varepsilon)\circ (\Psi\circ \Psi')=(\rho(\varepsilon)\circ \Psi)\circ (\rho(\varepsilon)\circ \Psi') 
\end{equation}
for every $\varepsilon\in \operatorname{Par}(P)$ and $\Psi,\Psi'\in\operatorname{Psd}(E)$, then the first part of (\ref{compatibility_action}) is also satisfied and for every fixed $\Psi$ the rule $\varepsilon\mapsto \rho(\varepsilon)\circ \Psi$ is homomorphic, so that Corollary \ref{particular_case} holds analogously.
\begin{example}
Recall that a ring $R$ is \textit{Boolean} if each element is idempotent, i.e., if $x*x=x$ for every $x\in R$. Let $\operatorname{Bol}(P)$ be a Boolean ring\footnote{We suspect that the same holds, more generally, for Von Neumann regular rings, but we do not have a proof of this.} and take  $\operatorname{Par}(P)= \operatorname{Bol}(P)\otimes_{\mathbb{Z}} \mathbb{C}$. Let $\rho:\operatorname{Bol}(P)\rightarrow \operatorname{End}_{\mathbb{C}}(\Gamma (E))$ be a faithful representation of this Boolean ring and notice that (\ref{boolean_condition}) is immediately satisfied (recall that every Boolean ring is commutative). Tensoring with $\mathbb{C}$ we get a faithful representation of $\operatorname{Par}(P)= \operatorname{Bol}(P)\otimes_{\mathbb{Z}} \mathbb{C}$. Finally, notice that every Boolean ring has square roots, since for every $x$ we have $x^2=x$, i.e., $\sqrt{x}=x$.
\end{example}
\begin{example}[a more concrete case]
Let $P=M\times A$ be a trivial algebra bundle, so that $\Gamma(P)\simeq C^{\infty}(M;A)$. Let $\operatorname{Bol}(A)\subset A$ be the Boolean ring of the idempotent elements of $A$ and take $\operatorname{Bol}(P)$ as the set of functions  $f:M\rightarrow A$ such that $f(x)\in\operatorname{Bol}(A)$ for every $x\in M$. Now, let $\rho:A \rightarrow \operatorname{End}(F)$ be a faithful representation of $A$ in the typical fiber of $E$. It induces a faithful representation $\rho:\Gamma(P)\rightarrow \operatorname{End}_{\mathbb{C}}(\Gamma(E))$ and, therefore, by restriction a faithful representation of $\operatorname{Par}(P)$.
\end{example}
\end{enumerate}

\section{Proof of Theorem \ref{main_theorem}}\label{sec_proof}

\quad\;\,In this section we prove our emergence theorem. The proof will be inductive on the number $r$ of variables of $S_2$. In order to prove the base case, i.e., the emergence theorem when $S_2$ is a univariate polynomial, we will need to use some additivity and multiplicativity properties of the emergence phenomena. In turn, the induction step will be based on a technical lemma. 

In order to better understand the whole proof, this section will be organized as follows. In Subsection \ref{sec_properties} we prove the basic properties of the emergence phenomena needed to prove the base step. In Subsection \ref{sec_base_induction} this base step is proved. In Subsection \ref{sec_theorem}, Theorem \ref{main_theorem} is finally demonstrated, with the technical lemma used for the induction step being presented before in Subsection \ref{step_5}.

\subsection{Properties of Emergence Phenomena}\label{sec_properties}

\begin{itemize}
    \item In the following discussion, when the degree of a GPT does not matter it will be made implicit in order to simplify the notation. In these cases we will also write $\varepsilon$  instead of $\varepsilon(\ell)$. Thus, from now on, by saying ``\textit{let $\Psi_{\varepsilon}$ be a GPT over $\operatorname{GB}(M)$}'' we mean that it is any GPT of any degree $\ell$. 
\end{itemize}

\begin{definition}
Let $\Psi_{\varepsilon}$ and $\Psi'_{\varepsilon'}$ be GPT over the same generalized background $\operatorname{GB}(M)$. The \textit{sum} and the \textit{composition} between them are the GPT over $\operatorname{GB}(M)$ given by $\Psi^+_{(\varepsilon,\varepsilon')}=\Psi_{\varepsilon}+\Psi'_{\varepsilon'}$ and $\Psi^{\circ}_{(\varepsilon,\varepsilon')}=\Psi_{\varepsilon}\circ\Psi'_{\varepsilon'}$. Notice that the sum and the composition between GPT of degrees $\ell$ and $\ell'$ has degree $\ell+\ell'$. 
\end{definition}

\begin{lem}
\label{lemma_4}Let $\Psi_{1,\varepsilon}$, $\Psi_{2,\delta}$ and
$\Psi_{3,\kappa}$ be three GPT over the same generalized background $\operatorname{GB}(M)$ with fundamental parameter algebra $\operatorname{Par}(P)^\ell$, where $\ell$ is the degree of the first GPT, such that:
\begin{enumerate}
\item $\Psi_{1,\varepsilon}$ is multiplicative;
\item $\Psi_{1,\varepsilon}$ emerges from both $\Psi_{2,\delta}$
and $\Psi_{3,\kappa}$.
\end{enumerate}
Then $\Psi_{1,\varepsilon}$ emerges from the compositions $S_{2,\delta}\circ S_{3,\kappa}$ and $S_{3,\kappa}\circ S_{2,\delta}$.
\end{lem}
\begin{proof}
From the
second hypothesis we conclude that $\Psi_{1,\varepsilon}=\Psi_{2,F(\varepsilon)}$ and 
$\Psi_{1,\varepsilon}=\Psi_{3,G(\varepsilon)}$
for certain functions $F,G$. Composing them and using the first hypothesis, we find
\[
\Psi_{1,\varepsilon^2}= \Psi_{1,\varepsilon}\circ\Psi_{1,\varepsilon}=\Psi_{2,F(\varepsilon)}\circ\Psi_{3,G(\varepsilon)}=\Psi_{3,G(\varepsilon)}\circ\Psi_{2,F(\varepsilon)}.
\]
Let $\sqrt{-}:\operatorname{Par}(P)^\ell\rightarrow \operatorname{Par}(P)^\ell$ be a function selecting to each fundamental parameter $\varepsilon'$ a square root $\sqrt{\varepsilon'}$, which exists by hypothesis. Then, for every $\varepsilon'$ one gets
\[
\Psi_{1,\varepsilon'}=\Psi_{2,F(\sqrt{\varepsilon'})}\circ\Psi_{3,G(\sqrt{\varepsilon'})}=\Psi_{3,G(\sqrt{\varepsilon'})}\circ\Psi_{2,F(\sqrt{\varepsilon'})} = \Psi^{\circ}_{H(\varepsilon')},
\]
finishing the proof.
\end{proof}
In a completely analogous way one proves the following.
\begin{lem}
\label{lemma_3} Let $\Psi_{1,\varepsilon}$, $\Psi_{2,\delta}$ and
$\Psi_{3,\kappa}$ be three GPT over the same generalized background $\operatorname{GB}(M)$, such that:
\begin{enumerate}
\item $\Psi_{1,\varepsilon}$ is scalar invariant;
\item $\Psi_{1,\varepsilon}$ emerges from both $\Psi_{2,\delta}$
and $\Psi_{3,\kappa}$.
\end{enumerate}
Then $\Psi_{1,\varepsilon}$ also emerges from the sum $\Psi_{2,\delta}+\Psi_{3,\kappa}$.
\end{lem}

\subsection{Base of Induction}\label{sec_base_induction}
$\quad\;\,$In this subsection, using the additivity and multiplicativity properties of last section, we will prove the following lemma, which will be the base of the induction step in the proof of Theorem \ref{main_theorem}:

\begin{lem}
\label{lemma_4_1} Let $\operatorname{GB}(M)$ be generalized background of $(\ell,k)$-type. Let $\Psi_{1,\varepsilon(\ell)}$ be a GPT of degree $\ell$ and let $\Psi_{2,\delta(\ell')}$ a PPT of degree $(l,\ell')$ in $r=1$ variables, defined on $\operatorname{GB}(M)$ and such that $\ell'=k'\ell$, with $0<k'\leq k$. Suppose that:
\begin{enumerate}
    \item $\Psi_{1,\varepsilon(\ell)}$ is homomorphic;
    \item $\Psi_{2,\delta(\ell')}$ is right-invertible and the coefficient functions $f_\alpha:\operatorname{Par}(P)^{k'\ell}\rightarrow \mathbb{R}$ of the polynomial $p^l_{\ell'}$  defining $\Psi_{2,\delta(\ell')}$  belongs to the functional calculus $C_{k'\ell;k}(P;\mathbb{K})$.
\end{enumerate}
Then $\Psi_{1,\varepsilon(\ell)}$ emerges from $\Psi_{2,\delta(\ell')}$.
\end{lem}

We begin with another lemma.

\begin{lem}
\label{lemma_2} Let $\Psi_{1,\varepsilon(\ell)}$ be a GPT of degree $\ell$ defined on a generalized background $\operatorname{GB}(M)$ of $(\ell,k)$-type, with $k>0$. Then $\Psi_{1,\varepsilon(\ell)}$ emerges from every GPT $\Psi_{2,\delta(\ell')}$ over $\operatorname{GB}(M)$, which has degree $\ell'= k'\ell$ for some $0<k'\leq k$, and such that   
$\Psi_{2,\delta(\ell')}^{l}=g(\delta(\ell'))\Psi^{l}$, with $l\geq0$,
where $\Psi$ is right-invertible and $g\in C_{k'\ell;k}(P;\mathbb{K})$.
\end{lem}
\begin{proof}
Since $\Psi^0=I$ is right-invertible,
the case $l=0$ is a particular setup of case $l=1$. Furthermore,
if $l>1$ and $\Psi$ is right-invertible, then $\Xi=\Psi^{l}$ is
right-invertible too, so that the case $l>1$ also follows from the
$l=1$ case. Thus, it is enough to work with $l=1$. Thus, let $R_{\Psi}\in\operatorname{Op}(E)$ be a right-inverse for $\Psi$ and notice that to find an emergence from $\Psi_{1,\varepsilon(\ell)}$ to $\Psi_{2,\delta(\ell')}$ is equivalent to building a function $F:\operatorname{Par}(P)^\ell \rightarrow \operatorname{Par}(P)^{\ell'}$ such that $\Psi_{1,\varepsilon(\ell)}\circ R_{\Psi}=g(F(\varepsilon(\ell)))I$. From (\ref{functional_calculus}) and from the fact that the right multiplication by right-invertible operators is injective, last condition is in turn equivalent to the existence of $F$ such that $\Psi^{\ell'}_g \circ \Psi_{1,\varepsilon(\ell)}\circ R_{\Psi}=F(\varepsilon(\ell))\cdot^{\ell'}I$, but this actually defines $F$ via $\operatorname{Par}(P)^{\ell'}\cdot^{\ell'}I\simeq \operatorname{Par}(P)^{\ell'}$.
\end{proof}

\begin{proof}[Sketch of proof of Lemma \ref{lemma_4_1}]
Given a PPT
$ 
\Psi_{2,\delta(\ell')}=\sum_if_i(\delta(\ell'))\Psi^i
$
in the hypothesis, for each $j=1,...,l$ let $\Gamma_{j}=\sum_{i=j}^{l}f_{i}(\delta(\ell'))\Psi^{i-1}$
and notice that 
\[
\Psi_{2,\delta(\ell')}=(\sum_{i=1}^{l}f_{i}(\delta(\ell'))\Psi^{i-1})\circ\Psi=\Gamma_{1}(\delta(\ell'))\circ\Psi.
\]
Since $\Psi$ is right-invertible and $\operatorname{GB}(M)$ is of $(\ell,k)$-type, with $k>0$, from Lemma \ref{lemma_2} it follows
that $\Psi_{1,\varepsilon(\ell)}$ emerges from $1\cdot\Psi$\footnote{Here we are using explicitly that the functional calculus is unital.}.
Thus, if $S_{\varepsilon(\ell)}$ itself emerges from $\Gamma_{1}$
one can use Lemma \ref{lemma_4} to conclude that it actually emerges
from $\Gamma_{1}\circ\Psi$. In turn, notice that $\Gamma_{1}=f_{1}\cdot I+\Gamma_{2}\circ\Psi=\Gamma_{2}\circ\Psi+f_{1}\cdot I$.
But, since $I$ is right-invertible and since $f_{1}\in C_{\ell';k}(P;\mathbb{K})$,
from Lemma \ref{lemma_2}
we get that $S_{\varepsilon(\ell)}$ emerges from the theory defined
by $f_{1}\cdot I$, while by the same argument we see that $\Gamma_{2}\circ\Psi$
emerges from $f_{1}\cdot I$. Therefore, if $\Psi_{1,\varepsilon(\ell)}$
emerges from $\Gamma_{2}\circ\Psi$ we will be able to use Lemma \ref{lemma_3}
to conclude that it emerges from $\Gamma_{1}$, finishing the proof.
It happens that, as done for $\Gamma_{1}\circ\Psi$, we see that $\Gamma_{2}$
emerges from $\Psi$ and we already know that $\Psi_{1,\varepsilon(\ell)}$
emerges from $\Psi$. Thus, our problem is to prove that $\Psi_{1,\varepsilon(\ell)}$
emerges from $\Gamma_{2}$ instead of from $\Gamma_{1}$. A finite
induction argument proves that if $\Psi_{1,\varepsilon(\ell)}$ emerges
from $\Gamma_{l}$, then it emerges from $\Gamma_{j}$, for each $j=1,...,l$.
Recall that $\Gamma_{l}=f_{l}\cdot\Psi^{l-1}$. Since $f_{l}\in C_{\ell';k}(P;\mathbb{K})$
we can use Lemma \ref{lemma_2} to see that $S_{\varepsilon(\ell)}$
really emerges from $\Gamma_{l}$.
\end{proof}

\subsection{Technical Lemma for Induction Step \label{step_5}}

$\quad\;\,$ Also as a consequence of the properties of the emergence phenomena, we can now prove the following technical lemma, which will be used in the induction step of Theorem \ref{main_theorem}.

\begin{lem}
\label{lemma_5} Let $\Psi_{1,\varepsilon}$ be a GPT over a generalized background $\operatorname{GB}(M)$. Given $l\geq1$, let $\Psi_{2_j,\delta_{j}}$ and $\Psi_{3_s,\kappa_{s}}$,
with $1\leq j,s\leq l$ be two families of GPT, also defined over  $\operatorname{GB}(M)$. Assume that:
\begin{enumerate}
\item $\Psi_{1,\varepsilon}$ is homomorphic;
\item $\Psi_{1,\varepsilon}$ emerges from $\Psi_{2_{j},\delta_{j}}$
and from $\Psi_{3_{s},\kappa_{s}}$ for every $j,s$.
\end{enumerate}
Then $\Psi_{1,\varepsilon}$ emerges from $\Psi_{\delta_{J},\kappa_{J}}^{s}=\sum_{j=1}^{s}\Psi_{2_{j},\delta_{j}}\circ \Psi_{3_{j},\kappa_{j}}$, for every $s=1,...,l$
\end{lem}
\begin{proof}
We proceed by induction in $l$. First of all, notice that from the
first two hypotheses and from Lemma \ref{lemma_4} we see that $\Psi_{1,\varepsilon}$
emerges from the composition $\Psi_{2_{j},\delta_{j}}\circ \Psi_{3_{j},\kappa_{j}}$
for every $j=1,...,l$. In particular, it emerges from $\Psi_{\delta_{1},\kappa_{1}}^{1}=\Psi_{2_{1},\delta_{1}}\circ \Psi_{3_{1},\kappa_{1}}$,
which is the base of induction. For every $m=1,...,l-1$ it also emerges
from $\Psi_{2_{m+1},\delta_{m+1}}\circ \Psi_{3_{m+1},\kappa_{m+1}}$. For
the induction step, suppose that $\Psi_{1,\varepsilon}$ emerges
from $\Psi_{\delta_{J},\kappa_{J}}^{m}=\sum_{j=1}^{m}\Psi_{2_{j},\delta_{j}}\circ \Psi_{3_{j},\kappa_{j}}$
for every $1\leq m\leq l-1$ and let us show that it emerges from
$\Psi_{\delta_{J},\kappa_{J}}^{m+1}$. Notice that 
\begin{eqnarray*}
\Psi_{\delta_{J},\kappa_{J}}^{m+1} & = & \sum_{j=1}^{m+1}\Psi_{2_{j},\delta_{j}}\circ \Psi_{3_{j},\kappa_{j}}=\sum_{j=1}^{m}(\Psi_{2_{j},\delta_{j}}\circ \Psi_{3_{j},\kappa_{j}})+\Psi_{2_{m+1},\delta_{m+1}}\circ \Psi_{3_{m+1},\kappa_{m+1}}\\
 & = & \Psi_{\delta_{J},\kappa_{J}}^{m}+(\Psi_{2_{m+1},\delta_{m+1}}\circ \Psi_{3_{m+1},\kappa_{m+1}}).
\end{eqnarray*}
From the induction hypothesis $\Psi_{1,\varepsilon}$ emerges from
$\Psi_{\delta_{J},\kappa_{J}}^{m}$, while by the above it also emerges
from $\Psi_{2_{m+1},\delta_{m+1}}\circ \Psi_{3_{m+1},\kappa_{m+1}}$. The result then follows from Lemma \ref{lemma_3}.
\end{proof}

\subsection{Proof of Theorem \ref{main_theorem} }\label{sec_theorem}

$\quad\;\,$We can finally proof our emergence theorem. For the convenience of the reader, we state it again.$\underset{\;}{\underset{\;}{\;}}$

\noindent \textbf{Theorem 3.1} (Emergence Theorem)
\textit{Let $M$ be a compact and oriented manifold and let $\operatorname{GB}(M)$ be a generalized background of $(\ell,k)$-type, with $k>0$. Let $\Psi_{1,\varepsilon(\ell)}$ be a GPT of degree $\ell$ and let $\Psi_{2,\delta(\ell)}$ be a PPT of degree $(l,\ell')$ in $r$ variables, where $\ell '=k'\ell$, with $0<k'\leq k$, defined on $\operatorname{GB}(M)$. Suppose that:
\begin{enumerate}
    \item $\Psi_{1,\varepsilon(\ell)}$ is homomorphic;
    \item $\Psi_{2,\delta(\ell)}$ is right-invertible and the coefficient functions $f_{\alpha}:\operatorname{Par}(P)^{k'\ell}\rightarrow \mathbb{K}$ belongs to the functional calculus $C_{k'\ell;k}(P;\mathbb{K})$. 
\end{enumerate}
Then $\Psi_{1,\varepsilon(\ell)}$ emerges from $\Psi_{2,\delta(\ell)}$.}

\begin{proof} The proof will be done by induction in $r$. The base of induction is Lemma \ref{lemma_4_1}. Suppose that the theorem holds
for each $r=1,...,q$ and let us show that it holds for $r=q+1$.
Let $p_{\ell';q+1}^{l}[x_{1},...,x_{r+1}]=\sum_{\vert\alpha\vert\leq l}f_{\alpha}\cdot x^{\alpha}$
be a multivariate polynomial with coefficients in $\operatorname{Map}(\operatorname{Par}(P)^{\ell'};\mathbb{K})$, which actually belong
to $C_{k'\ell;k}(P;\mathbb{K})$. Since for every commutative ring $R$ we have $R[x_{1},...,x_{q+1}]\simeq R[x_{1},...,x_{q}][x_{q+1}]$, given right-invertible generalized operators $\Psi_1,...,\Psi_{q+1}\in \operatorname{Op}(E)$ one can write
$$
\Psi_{2,\delta(\ell)}=p_{\ell';q+1}^{l}[\Psi_{1},...,\Psi_{r+1}]=\sum_{j}p^{l_j}_{\ell';q,j}[\Psi_{1},...,\Psi_{q}]\cdot \Psi_{q+1}^{j},
$$
where each $p^{l_j}_{\ell';q,j}[x_1,...,x_q]\in\operatorname{Map}_{l_j}(\operatorname{Par}(P)^{\ell'};\mathbb{K})[x_1,...,x_q]$ has coefficients which belongs to  belongs to $C_{k'\ell;k}(P;\mathbb{K})$. Thus, by the induction hypothesis, $\Psi_{1,\varepsilon(\ell)}$ emerges from $p^{l_j}_{\ell';q,j}[x_1,...,x_q]$. Since $\Psi_{q+1}$ is right-invertible and since the functional calculus is unital, from Lemma  \ref{lemma_2} we see that $\Psi_{1,\varepsilon(\ell)}$ emerges from $\Psi^j_{q+1}$. The result then follows from Lemma \ref{lemma_5}.
\end{proof}

\section*{Acknowledgments}

The first author was supported by CAPES (grant number 88887.187703/2018-00).

\bibliographystyle{plainnat}
\bibliography{emergence}

\end{document}